\documentclass[letterpaper]{article} % DO NOT CHANGE THIS
\usepackage{aaai2026}  % DO NOT CHANGE THIS
\usepackage{times}  % DO NOT CHANGE THIS
\usepackage{helvet}  % DO NOT CHANGE THIS
\usepackage{courier}  % DO NOT CHANGE THIS
\usepackage[hyphens]{url}  % DO NOT CHANGE THIS
\usepackage{graphicx} % DO NOT CHANGE THIS
\usepackage{natbib}  % DO NOT CHANGE THIS AND DO NOT ADD ANY OPTIONS TO IT
\usepackage{caption} % DO NOT CHANGE THIS AND DO NOT ADD ANY OPTIONS TO IT
\usepackage[ruled,vlined,linesnumbered]{algorithm2e}
\usepackage{algpseudocode}

\usepackage{amsthm}
\usepackage{amsmath}
\usepackage{cleveref}

\usepackage[modulo,switch]{lineno} % module=every 5, switch = 2 columns left, right 
\newtheorem{theorem}{Theorem}

\theoremstyle{definition}
\newtheorem{definition}{Definition}
\newtheorem{property}{Property}

\newtheorem{lemma}{Lemma}

\usepackage{amssymb}

\title{BTPG-max: Achieving Local Maximal Bidirectional Pairs for Bidirectional
Temporal Plan Graphs}
\author{
    Yifan Su,
    Rishi Veerapaneni,
    Jiaoyang Li
}
\affiliations{
    Carnegie Mellon University\\
    
    \{yifansu, rveerapa\}@andrew.cmu.edu, jiaoyangli@cmu.edu
}

\usepackage{bibentry}
\begin{document}

\maketitle

\begin{abstract}
Multi-Agent Path Finding (MAPF) requires computing collision-free paths for multiple agents in shared environment. Most MAPF planners assume that each agent reaches a specific location at a specific timestep, but this is infeasible to directly follow on real systems where delays often occur. To address collisions caused by agents deviating due to delays, the Temporal Plan Graph (TPG) was proposed, which converts a MAPF time dependent solution into a time independent set of inter-agent dependencies. Recently, a Bidirectional TPG (BTPG) was proposed which relaxed some dependencies into ``bidirectional pairs" and improved efficiency of agents executing their MAPF solution with delays. Our work improves upon this prior work by designing an algorithm, BPTG-max, that finds more bidirectional pairs. Our main theoretical contribution is in designing the BTPG-max algorithm is locally optimal, i.e. which constructs a BTPG where no additional bidirectional pairs can be added. We also show how in practice BTPG-max leads to BTPGs with significantly more bidirectional edges, superior anytime behavior, and improves robustness to delays. 
\end{abstract}

% Uncomment the following to link to your code, datasets, an extended version or similar.
% You must keep this block between (not within) the abstract and the main body of the paper.
% \begin{links}
%     \link{Code}{https://aaai.org/example/code}
%     \link{Datasets}{https://aaai.org/example/datasets}
%     \link{Extended version}{https://aaai.org/example/extended-version}
% \end{links}

\section{Introduction}

Multi-Agent Path Finding (MAPF) is the problem of finding collision-free paths for a set of agents in a shared workspace. MAPF is one core component of intelligent multi-agent teams and has direct applications in systems with many robots like warehouse management. 

Solving an MAPF instance, in theory, gives a MAPF solution which is a set of paths for robots to follow. 
This assumes that agents can strictly follow the paths, requiring that they arrive at specific locations at specific timesteps.
However, in practice, this is impossible as robots may be delayed during execution (e.g., due to kinematic constraints or wheel slippage). If agents naively follow their original paths without accounting for delays and other temporal differences, they may collide with other agents or get stuck in deadlocks.

To handle temporal differences in execution, Hönig et al. \shortcite{hoenig_multi-agent_2016_tpg} introduced a Temporal Plan Graph (TPG). 
The main innovation is that a MAPF solution can be processed into a set of dependencies. 
In particular, agents are only dependent on each other when their paths spatially intersect (e.g., location E in \Cref{fig:TPG and BTPG}). At these intersections, denoted as conflict locations, the inter-agent dependency is that agents must maintain their relative passing order specified in the MAPF solution. For example, in \Cref{fig:TPG and BTPG}, the green agent should pass location E first. During execution, agents can go as fast/slow as long as they satisfy this inter-agent dependency. The TPG paper proves how following these dependencies, regardless of time, will still lead to collision-free and deadlock-free execution.

However, a TPG can be overly restrictive. In \Cref{fig:TPG and BTPG}, during execution, if the green agent encounters a delay, the TPG would make the blue agent wait for the green agent to pass through location $E$ first. Consequently, a delay of the green agent also leads to a delay for the blue agent. To address this issue, the Bidirectional Temporal Plan Graph (BTPG) \cite{su2024bidirectional_tpg} was introduced, allowing the blue agent to choose to pass through $E$ first in such situations to avoid unnecessary waiting.

\begin{figure}[t]
    \centering
\includegraphics[width=1.0\linewidth]{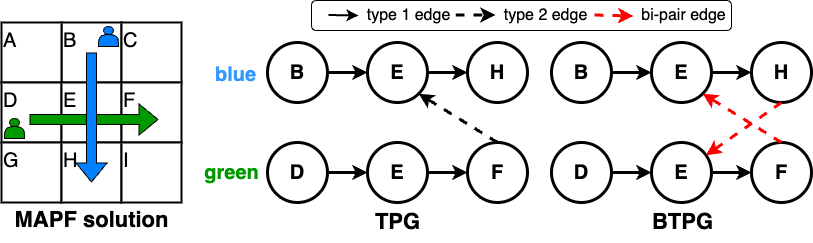}
    \caption{Example of TPG and BTPG for a MAPF solution that requires the green agent to pass through E first. The TPG requires blue to wait for green (even if green gets delays) while the BTPG allows either agent to cross the intersection (thus leading to more robustness to delays).}
    \label{fig:TPG and BTPG}
    \vspace{-1.2em}
\end{figure}

BTPG introduces \emph{bidirectional pairs} (or “bi‑pairs”), which are pairs of opposite dependency edges at a conflict location that allow either agent to go first without causing deadlock. In the original BTPG construction, each dependency in the TPG is tested for reversibility: if reversing an edge does not introduce a deadlock, then the agent can safely pass through the conflict point first—even if the MAPF solution originally scheduled it to go later. These reversible edges form bi‑pairs. Although any cycle in a TPG indicates deadlock, the BTPG paper observed that some cycles in the BTPG graph would not lead to deadlock. Leveraging this, they proposed \emph{BTPG‑o}, which applies hand‑crafted rules to detect and exclude such non‑deadlocking cycles to find more bi‑pairs. 

In this work, we point out that BTPG‑o still misses certain non‑deadlock cycle types.
To address this, we propose \emph{BTPG‑max}, which identifies \emph{all} non‑deadlock cycles to discover the locally maximal number of bi‑pairs. Our experiments show that BTPG‑max finds more bi‑pairs, yielding faster execution times and markedly better anytime performance.

\section{Preliminary}

\subsection{Problem Formulation}
The MAPF problem takes as input a graph and a set of start and goal locations for a group of agents. The MAPF solver provides each agent with a collision-free path, allowing them to move simultaneously to their goal locations on the graph. 
In the standard MAPF set-up, time is discretized into timesteps, with solution paths dictating that agents need to be at specific locations at specific timesteps. However, in actual execution, robots cannot meet this requirement due to hardware or communication issues that can cause delays. If agents followed the paths without adjusting for delays, they could collide with each other.
We are thus interested in determining how agents should handle delays during execution. 

Our work directly builds on top of TPG \cite{hoenig_multi-agent_2016_tpg} and BTPG \cite{su2024bidirectional_tpg}. Constructing a TPG or BTPG to handle delays is unique in that it is a post-processing technique on top of the MAPF solution that is done \textit{before} execution. This eliminates the need to replan during execution (avoiding additional delays caused by extra computation during execution), as agents only need to follow the dependencies specified by the (B)TPG.
% \rishi{The rest of this section describes TPG and BTPG, but readers completely unfamiliar with both TPG and BTPG are encouraged to read \cite{su2024bidirectional_tpg} for a more in-depth background due to page limits.}

\subsection{Temporal Plan Graph (TPG)}
% Because we also allow the following conflicts, we will use the definition of TPG that allows the following conflicts in our paper, rather than the most original definition.
As discussed in the introduction, TPG is a way to interpret a MAPF solution as a set of dependencies \cite{hoenig_multi-agent_2016_tpg}. We now formally define it.

\begin{definition}[TPG]
TPG is a directed acyclic graph $G=(V,E)$, where each vertex $v_i^m \in V$ corresponds to a state of agent $m$ being at location loc($v_i^m$). The index $i$ indicates that loc($v_i^m$) is the $i$-th location on the path of agent $m$. The edges in $E$ define the precedence dependencies between states and are categorized into two types: type-1 edges $E_1$ and type-2 edges $E_2$. A type-1 edge $(v_i^m, v_{i+1}^m)$ enforces that agent $m$ must reach loc($v_{i}^m$) before moving to loc($v_{i+1}^m$) . A type-2 edge $(v_i^m, v_j^n)$ enforces that agent $n$ can only reach loc($v_j^n$) after or at the same time as agent $m$ reaches loc($v_i^m)$.
\end{definition}

Conceptually, if all agents had paths that never intersected, the TPG would consist of only type-1 edges. A type-2 edge $(v_i^m, v_j^n)$ occurs when the paths of both agents $n$ and $m$ need to pass through conflict location loc($v_j^n$)$=$loc($v_{i-1}^m$) and $m$ should pass it first.

\subsubsection{TPG execution policy}
To move to the next state, an agent must satisfy all dependencies represented by all edges pointing to that state. 
%A type-1 edge requires the agent to follow the path specified by the MAPF solution.
%A type-2 edge $(v_i^m, v_j^n)$ requires that during execution agent $n$ can only enter loc($v_j^n)$ when agent $m$ has already visited loc($v_i^m$) or is entering it at the same timestep.

For example, the TPG in \Cref{fig:TPG and BTPG} has a type-2 edge from $F$ to $E$. During execution, the blue agent can only enter $E$ once the $F \rightarrow E$ dependency is satisfied. This means that blue can only enter $E$ after green is entering or has already entered $F$.

By following a TPG during execution, we allow agents to travel more flexibly, as they are not tied to specific timesteps or intervals. Even if agents experience delays, following the TPG ensures that the passing order at each conflict location remains consistent with the MAPF solution. This enables them to execute the MAPF solution without collisions, deadlocks, or replanning when delays occur.

\subsection{Bidirectional Temporal Plan Graph (BTPG)}
A TPG requires agents to strictly follow the passing order at each location as specified by the MAPF solution. When delays occur, this can cause agents to unnecessarily wait and reduce efficiency. 
Bidirectional Temporal Plan Graphs (BTPGs) were proposed to address this issue \cite{su2024bidirectional_tpg}. Conceptually, at certain intersections where it is safe, e.g. \Cref{fig:TPG and BTPG}, BTPG replaces the fixed ordering and allows either of the agents to cross first. In the example, this means that either blue or green can cross the intersection first. This involves changing certain type-2 edges into bidirectional pairs and ensuring that the resultant BTPG is deadlock-free.

\begin{definition}[BTPG]
    Building on TPG, BTPG introduces the concept of bidirectional pairs $E_{pair}$. A bidirectional pair (bi-pair) consists of a pair of type-2 edges $\{(v_i^m, v_j^n), (v_{j+1}^n, v_{i-1}^m)\}$ with $\text{loc}(v_j^n) = \text{loc}(v_{i-1}^m)$ (denoted as conflict location). These two edges represent both options of agent $m$ going before $n$ and $n$ going before $m$ at the conflict location, respectively.
\end{definition}

\subsubsection{BTPG execution policy}
The BTPG execution policy is similar to the TPG execution policy but with a key difference. When agents encounter a bidirectional pair, only one edge is selected based on a ``first-come, first-served" rule. This means that either agent in the pair can be the first to enter the conflict location. When the first agent arrives, the edge that allows this agent to enter first is selected and the other edge is deleted. For example, in the BTPG of \Cref{fig:TPG and BTPG}, the red type-2 edges form a bi-pair. Now during execution, if green gets delayed, blue can select the $H \rightarrow E$ edge and enter $E$ first, and the edge $F \rightarrow E$ is discarded. This avoids the unnecessary waiting that would have been required if we had followed the original TPG.

Importantly, an edge $(v_{i}^m, v_j^n)$ from a bi-pair will only be selected during execution when agent $m$ reaches $v_{i-1}^m$ first. This is crucial for additional BTPG optimizations.

\subsubsection{BTPG construction}
The method for building a BTPG starts with a TPG and is described in \Cref{algo:Constructing BTPG - baseline} without the blue lines. Each type-2 edge is checked one by one to determine whether it can be reversed to form a bi-pair (\Cref{line:for-loop}, \Cref{algo:Constructing BTPG - baseline}). If reversing the edge doesn't cause a deadlock, the agents could reverse passing orders during execution. 
In a TPG, the presence of a cycle leads to a deadlock~\cite{berndt_feedback_2020, coskun_deadlock-free_2021}.
Thus the way to check for a deadlock is to see if the reversed edge forms a cycle (\Cref{line:call-hasCycle}, \Cref{algo:Constructing BTPG - baseline}).
If not, the type-2 edge is converted into a bi-pair. 

\subsection{BTPG-optimized (BTPG-o)}
However, due to the existence of bi-pairs in a BTPG, agents choose certain type-2 edges during execution, so some bi-pair edges in the cycle may provabely never be selected to execute. This means that certain cycles will not be encountered during execution and will not lead to a deadlock, and are thus called ``non-deadlock cycles".

\begin{property}[Non-deadlock cycles in BTPG-o]
\label{prop:BTPG-o}
 If a cycle contains a vertex $v_i^n$ and a bi-pair edge that points from $v_j^n, j>i$, then this cycle will not lead to a deadlock, and we call it a non-deadlock cycle.
\end{property}

Conceptually, if the cycle involving $v_i^n$ can cause a deadlock, agent $n$ will get stuck at $v_{i-1}^n$ (or even earlier), preventing it from reaching all subsequent vertices $v_j^n, j \ge i$. So, according to our BTPG execution policy, the bi-pair edge that originates from $v_j^n, j>i$ can not be selected as $n$ cannot reach $v_{j-1}^n$. Hence, this cycle will never be encountered and will not cause a deadlock during execution. For example, the cylce $v_d^3 \rightarrow v_b^2 \rightarrow \cdots \rightarrow v_{c-1}^2 \rightarrow v_c^2 \rightarrow v_a^1 \rightarrow v_f^4 \rightarrow v_d^3$ is a non-deadlock cycle, because the edge $(v_c^2, v_a^1)$ is a bi-pair edge and $v_b^2, b<c$ is in the cycle. The BTPG-o algorithm uses this property and excludes these non-deadlock cycles, allowing it to find more bi-pairs that make the result BTPG more flexible.

% \subsection{Summary of TPG and BTPG}
% We have discussed how TPGs encode dependencies between agents as type-2 edges. We then described BTPGs with the introduction of bi-pair edges. During execution, the direction of a bi-pair is decided by the first agent that reaches the conflict location, which allows agents to pass through the conflict locations via a ``first-come first-served" policy. 

% A BTPG is constructed by checking if reversing a type-2 edge in a TPG is deadlock-free. Prior work has shown that cycles in TPGs indicate deadlocks. However, BTPG-o observes that a certain type of cycles in BTPGs will not lead to deadlocks due to the BTPG execution policy. Thus, more type-2 edges can be converted into bi-pairs be ignoring these non-deadlock cycles.

In our work, we identify more types of cycles that do not lead to deadlocks. Section \ref{sec:non-deadlock-cycle-btpg} discusses the complete set of \emph{all} types of non-deadlock cycles. Section \ref{sec:btpg-max} then describes our BTPG-max algorithm, 
which is based on this insight and finds the locally maximal number of bi-pairs by ignoring all non-deadlock cycles. 
% for finding bi-pair edges, which requires finding deadlock cycles (i.e. requires cycle detection while ignoring these non-deadlock cycles).

\begin{figure}[t!]
    \centering
    \includegraphics[width=0.9\linewidth]{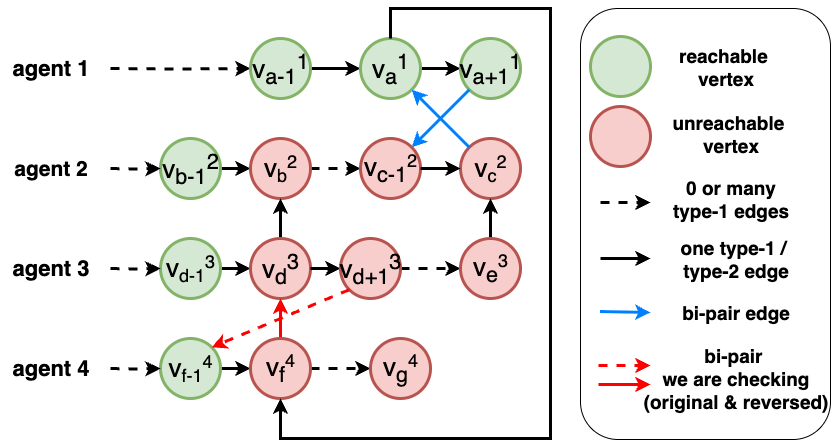}
    \caption{Example of new types of non-deadlock cycle discussed in Section \ref{sec:non-deadlock-cycle-btpg}.}
    \label{fig: new non-deadlock cycle}
    \vspace{-1em}
\end{figure}

\subsection{Other Related Works}

Recent works have explored online switching of TPG (Temporal Plan Graph) edge dependencies to handle delays, leveraging the key fact that cycles in a TPG imply deadlocks \cite{berndt_feedback_2020, coskun_deadlock-free_2021}. These methods aim to switch dependencies without introducing cycles. Switchable-Edge Search (SES) \cite{feng2024ses} uses heuristic search over type-2 edges to minimize expected cost based on current agent positions and delays. Berndt et al. \shortcite{berndt2023receding} propose a similar Switchable Action Dependency Graph and solve for switchable edges using receding horizon planning formulated as a Mixed-Integer Linear Programming problem. Liu et al. \shortcite{liu2024greedy_ldg_swap} greedily adjust edges for blocked agents to maximize agent progress, ensuring no cycles via DFS checks. Another approach re-solves MAPF while constraining agents to original paths \cite{kottinger2024delays_acid}, which scales for isolated delays but struggles with multiple. 

Separately, some methods plan delay-robust MAPF solutions, either explicitly allowing up to $K$ delays \cite{atzmon_robust_2021, chen2021robust_symmetry} or incorporating delay probabilities \cite{ma_multi-agent_2017}, though these require prior delay models. Space-Level CBS \cite{wagner_minimizing_2022} minimizes agent wait events, while Space-Order CBS \cite{wu2024spacetimespaceorderdirectlyplanning} directly computes TPGs to reduce type-2 edges.

% A separate line of research has focused on developing MAPF solutions that are themselves robust to delays. Some works explicitly plan MAPF solutions that are robust to $K$ delays \cite{atzmon_robust_2021,chen2021robust_symmetry} while some work \cite{ma_multi-agent_2017} incorporates delay probabilities. However, these require prior knowledge of the delay distribution. Space-Level CBS \cite{wagner_minimizing_2022} minimizes the number of instances agents need to wait for each other, but requires all agents to wait for each other at those instances. Space-Order CBS \cite{wu2024spacetimespaceorderdirectlyplanning} computes a TPG directly (instead of from a MAPF solution) and focuses on reducing the number of total type-2 edges.
Our method is orthogonal to these methods and could be used on top of them.

\section{Non-deadlock Cycles in BTPG} \label{sec:non-deadlock-cycle-btpg}
% To find more bi-pairs without encountering false positives, we need a new algorithm to identify all cycles that do not cause deadlocks. Therefore, we define non-deadlock cycles here in a more general way.
% \begin{definition}[Deadlock and Non-deadlock cycle in BTPG]
%     \label{def: non-deadlock cycle}
%     A deadlock cycle in BTPG is a cycle of vertices that may lead to a deadlock during execution. A non-deadlock cycle in BTPG is a cycle of vertices that will never lead to deadlock during execution. 
% \end{definition}

Our objective in this section is to find all \emph{non-deadlock cycles} in BTPG. These are a cycle of vertices and edges in a BTPG that will never lead to deadlocks during execution under the BTPG execution policy. Analogously, a \emph{deadlock cycle} in a BTPG is a cycle of vertices and edges that has a non-zero probability to lead to a deadlock during execution.

% In the analysis of \Cref{prop:BTPG-o}, we determined whether a cycle can lead to a deadlock by analyzing whether an agent is ``reachable" to the conflict locations of bi-pairs. Our work adopts this idea: if an agent cannot reach the conflict locations, it cannot choose the corresponding bi-pair edges during execution, making it a non-deadlock cycle. Here, we will define the reachable and unreachable vertex in detail. 
In the analysis of \Cref{prop:BTPG-o}, we claimed a cycle to be non-deadlock if a bi-pair edge in the cycle cannot be selected because the corresponding conflict location is ``unreachable" by the deadlocked agent. Our work expands this idea: we will classify all vertices in BTPG as reachable and unreachable and denote those cycles that have bi-pair edges with unreachable conflict locations as non-deadlock.

\begin{definition}[Reachable and unreachable vertex]
\label{def: reachable vertex}
    When a deadlock occurs, all vertices where agents might be located are reachable vertices, and all vertices where they cannot be located are unreachable vertices.
\end{definition}

Conceptually, reachable vertices are those before deadlocks. Unreachable vertices are those after deadlocks (since the agent is in a deadlock, it cannot reach future vertices) and, importantly, vertices who are dependent via type-1 and type-2 edges on other unreachable vertices.
For example, consider the cycle $v_d^3 \rightarrow v_b^2 \rightarrow \cdots \rightarrow v_{c-1}^2 \rightarrow v_c^2 \rightarrow v_a^1 \rightarrow v_f^4 \rightarrow v_d^3$ in \Cref{fig: new non-deadlock cycle}. If this cycle leads to a deadlock, then agent 2 can only reach up to vertex $v_{b-1}^2$. Agent 2 cannot reach future vertices like $v_{c-1}^2$. So, $v_{c-1}^2$ is unreachable. 

It is evident that if a vertex is in a cycle that constitutes a deadlock, then this vertex is unreachable. Consequently, all subsequent type-1 edges and normal type-2 edges cannot be satisfied. This will in turn cause those vertices to be unreachable, and we can cascade the logic down the graph. 

\begin{lemma}
\label{lem: unreachable}
    During execution, if a cycle leads to a deadlock, vertex $v_i^n$ is unreachable if and only if it is in the cycle or there is a path from a vertex $v_j^m$ in the cycle to it only through type-1 edges and normal type-2 edges (i.e., not type-2 edges in bi-pairs).
\end{lemma}
\begin{proof}
For the first case, because $v_i^n$ is in the deadlock cycle, it is definitely unreachable (otherwise the agent would not be in deadlock there). For the second case, since $v_j^m$ is in the cycle and thus unreachable, all type-1 edges and normal type-2 edges along the path from $v_j^m$ to $v_i^n$ cannot be satisfied, indicating that $v_i^n$ is also unreachable. 

We prove the other direction by contradiction: Assume that the unreachable vertex $v_i^n$ is not in the cycle and that there is no such path. Then, either (1) there is no path from a vertex in the cycle to  $v_i^n$, or (2) there is a path passing through some bi-pair edges from a vertex in the cycle to  $v_i^n$. 
In case 1, all edges pointing to vertex $v_i^n$ can be satisfied, meaning that $v_i^n$ is indeed reachable.
In case 2, based on the BTPG execution policy, the other edge in the bi-pair could be selected, in which case the current edge would be discarded, making that path disappear. 
Both cases contradict the assumption that $v_i^n$ is unreachable.
\end{proof}

Next, we determine whether the conflict vertex $v_j^m$ of a bi-pair edge within the cycle is unreachable. If it is, this bi-pair edge will not be selected, resulting in a non-deadlock cycle.

\begin{theorem}
\label{the: non-deadlock cycle}
    A cycle in BTPG is a non-deadlock cycle if and only if it contains a bi-pair edge $(v_i^m, v_j^n)$ such that $v_{i-1}^m$ is unreachable.
\end{theorem}
\begin{proof}
We first prove by contradiction that if the cycle satisfies the given conditions, it is a non-deadlock cycle. Assume that this cycle could lead to a deadlock. According to the BTPG execution policy, since $(v_{i-1}^m)$ is unreachable, the bi-pair edge could not be selected, meaning that the cycle would never be encountered during execution, leading to a contradiction. Thus, the cycle must be a non-deadlock cycle.
 
We then prove the other direction also by contradiction: Assume there is a non-deadlock cycle that does not satisfy the given conditions, meaning that all edges in this cycle can be executed (i.e., all bi-pair edges in the cycle can be selected to execute). Agents may encounter this cycle during execution, leading to a deadlock.
\end{proof}

Based on \Cref{the: non-deadlock cycle}, we can find new types of non-deadlock cycles that BTPG-o does not cover. For example, in \Cref{fig: new non-deadlock cycle} when evaluating if edge $(v_{d+1}^3, v_{f-1}^4)$ can be reversed to $(v_{f}^4, v_{d}^3)$, BTPG-o would classify the cycle $v_d^3 \rightarrow \cdots \rightarrow v_e^3 \rightarrow v_c^2 \rightarrow v_a^1 \rightarrow v_f^4 \rightarrow v_d^3$ as a deadlock cycle based on \Cref{prop:BTPG-o}. 
However, since $v_d^3$ is an unreachable vertex and there is a path from $v_d^3$ to $v_{c-1}^2$ only through type-1 and normal type-2 edges, $v_{c-1}^2$ is unreachable, and this cycle is a non-deadlock cycle as the bi-pair edge originating from $v_{c}^2$ cannot be chosen during execution.
% However, since there is a path from $v_i^3$ to $v_{g-1}^2$ only through type-1 and normal type-2 edges and $v_i^3$ is in the cycle, this cycle is a non-deadlock cycle.

%Section \ref{sec:btpg-max} creates an algorithm based on Theorem 1 that excludes all non-deadlock cycles when determining if type-2 edges can be converted into bi-pairs. Since the new algorithm excludes all non-deadlock cycles, it ensures that our deadlock cycle detection only prevents actual possible deadlocks and enables us to maximize bi-pairs found.

% Now we can construct an algorithm, based on Theorem 1, to eliminate all non-deadlock cycles when determine if type-2 edges can be converted into bi-pairs. Since the new algorithm can exclude all non-deadlock cycles, there will be no false positives, thereby ensuring that all bi-pairs are identified.

\subsection{Reachable deadlock cycle}
So far we have discussed non-deadlock cycles and how they can be identified. This implicitly defines deadlock cycles (i.e. those cycles which are not non-deadlock cycles). However, we can define an even stricter subset of deadlock cycles that are faster to detect.

In a deadlock situation, no agent can proceed—each gets stuck at the vertex just before entering the deadlock cycle. We refer to such cycles as reachable deadlock cycles.

\begin{definition}[Reachable deadlock cycle in BTPG]
    \label{def: reachable deadlock cycle}
    A reachable deadlock cycle in BTPG is a cycle such that for every (normal or bi-pair) type-2 edge $(v_i^m, v_j^n)$, $v_{j-1}^n$ is where robot $n$ gets stuck during execution.
\end{definition}
This means that when a deadlock occurs, all the agents involved in the deadlock will stop at the vertices just one before the deadlock.
So in Definition~\ref{def: reachable deadlock cycle}, if any vertex like \( v_{j-1}^n \) is marked as unreachable, then the cycle is not a reachable deadlock cycle because stuck agents cannot be at those locations when deadlock happens (based on Definition~\ref{def: reachable vertex}). 
 With this stricter property of a deadlock cycle, we can speed up the cycle detection process. 
 
Back to the previous example, when we check the reversed edge $(v_f^4, v_d^3)$ and traverse the path $v_d^3 \rightarrow \cdots \rightarrow v_e^3 \nrightarrow v_c^2$ during cycle detection, we can skip the edge $(v_e^3, v_c^2)$. Because if $v_d^3$ is in the deadlock cycle, there is a path from $v_d^3$ to $v_{c-1}^2$ through only the type-1 and normal type-2 edges, then $v_{c-1}^2$ is unreachable. $v_{c-1}^2$ is unreachable means that $v_c^2$ cannot be in a reachable deadlock cycle, so we can skip the edge $(v_e^3, v_c^2)$.
%for the cycle $v_d^3 \rightarrow v_e^3 \rightarrow v_c^2 \rightarrow v_a^1 \rightarrow v_f^4 \rightarrow v_d^3$ shown in \Cref{fig: new non-deadlock cycle}, if this is a reachable deadlock cycle, we should expect the four agents to stop at the nodes $v_{a-1}^1$, $v_{c-1}^2$, $v_{d-1}^3$, and $v_{f-1}^4$ according to \Cref{def: reachable deadlock cycle}. 
%However, since there is a path from $v_d^3$ to $v_{c-1}^2$ through only the type-1 and normal type-2 edges, $v_{c-1}^2$ is unreachable. Then, this cycle is not a reachable deadlock cycle (Definition~\ref{def: reachable deadlock cycle}). Therefore, we can skip the edge $(v_e^3, v_c^2)$ when traversing the path $v_d^3 \rightarrow v_e^3 \nrightarrow v_c^2$ during the cycle detection process. In short, traversing the edge $(v_e^3, v_c^2)$ would imply that agent 2 is stuck at $v_{c-1}^2$, which is not possible since $v_{c-1}^2$ is unreachable.
This operation will not affect the number of bi-pairs found but can accelerate the algorithm. 

Every possible deadlock has a corresponding reachable deadlock cycle, so we can design an algorithm that identifies only reachable deadlock cycles to determine whether a type-2 edge can be reversed.

% \rishi{To summarize, we have defined a ''reachable deadlock cycle" and shown how all possible deadlock cycles (no matter how complicated they are) must have a corresponding reachable deadlock cycle. Importantly, due to Theorem 1 defining an iff of non-deadlock cycles, by checking only for ``reachable deadlock cycles", we are checking for only the possible cycles that could occur and no extra cycles (compared to the original BTPG paper). (Improve this part, basically summarize the theoretical contribution here)}  

% \vspace*{-1.0em}
\begin{algorithm}[t!]
        \caption{BTPG-max: \textcolor{blue}{Blue} highlights the main difference from BTPG-o.}
        
        \label{algo:Constructing BTPG - baseline}
        \SetKwFor{While}{\fbox{while $E_{pair}$ has been updated do}}{}
        % \SetKwFunction{MyFunction}{MyFunction}
        
        \SetKwProg{Fn}{Function}{:}{end}
        \KwIn{TPG $G=(V, E_1 \cup E_2)$} 
        \KwOut{BTPG $\mathcal{G}$}
        \LinesNumbered
        \texttt{$E_{pair} \gets \emptyset$}\tcp*{set of bi-pair edges}
        \textcolor{blue}{\texttt{$E_{group} \gets$ Grouping($E_2$)}}\;
        % \While{\label{line:Epair-updated}}
        % {
            \For{\texttt{g} \textbf{in} \texttt{$E_{group}$} (or \textbf{until} \texttt{TimeOut}) \label{line:group for-loop}}{
                $\tilde{g} \gets \{(v_{j+1}^m, v_{i}^n)| (v_{i+1}^n, v_{j}^m) \in g\}$\;
                $E_{pair} \gets E_{pair} \cup \{g, \tilde{g}\}$\;
                % $E_{group} \gets E_{group} \setminus \{g\}$\;
                \For{\texttt{$e=(v_{j+1}^m, v_{i}^n)$} \textbf{in} \texttt{$\tilde{g}$} \label{line:for-loop}}
                {
                    $\mathcal{G}\gets(V, E_{group} \cup E_{pair}, E_1\cup E_2)$\;
                    \texttt{\textcolor{blue}{StateStatus $\gets \{\text{'reachable'}|\forall v \in V$\}}\label{lin:initial}}\;
                    \textcolor{blue}{\texttt{ UpdateStateStatus}(\texttt{StateStatus}, $\mathcal{G}$, $v_i^n$)\;}\label{lin:update}
                    \If{\texttt{DeadlockCycleDetection}($\mathcal{G}$, $v_{i}^n$, $v_{j+1}^m$, $\{e\}$, \texttt{\textcolor{blue}{StateStatus}})\label{line:call-hasCycle}}
                    { 
                        $E_{pair} \gets E_{pair} \setminus \{g,\tilde{g}\}$\;
                        % $E_{group} \gets E_{group} \cup \{g\}$\;
                        % \Break
                        \textbf{break}\;
                    }
                }
            }
        % }
        
        \Return $\mathcal{G}=(V, E_1 \cup E_2\cup E_{pair})$\;
\end{algorithm}

\begin{algorithm}[t!]

        \caption{DeadlockCycleDetection: \textcolor{blue}{Blue} highlights the main difference from BTPG-o.}
        \label{algo: hasCycle}
        \SetKwFunction{MyFunction}{MyFunction}
        \SetKwFor{mIf}{\textcolor{blue}{if}}{\textcolor{blue}{then}}{\textcolor{blue}{end}}
        \SetKwProg{Fn}{Function}{:}{end}
        \KwIn{(1) BTPG $\mathcal{G} = (V, E_{group} \cup E_{pair}, E_1\cup E_2)$, (2) current vertex for expansion $v_{i}^n$, (3) origin vertex $v_o$, (4) set(s) of edges $E_{vis}$ along the current DFS branch, (5) the set of status of vertices along the current DFS branch.} 
        \KwOut{\texttt{true} or \texttt{false}}
        \LinesNumbered
  
        % \texttt{$visited \gets visited \cup \{v_{i}^m\}$}\; 
        \If{$v_i^m = v_o$\label{line:cycle-found}}
                {
                \uIf{$E_{vis} \subseteq E_2$ \textbf{and} $ |E_{vis}| > 2 $\label{line:detect-rotation-cycle}}{\Return \texttt{false}\label{line:rotation-cycle}\;}
                \lElse{\Return \texttt{true}\label{line:non-rotation-cycle}}
                }
            \For{\texttt{$v_j^n$} \textbf{in} $\{v_j^m \in V \mid (v_i^n, v_j^m) \in E_1 \cup E_2 $\}} 
            {
                % \lIf{\texttt{$v_j^n \in visited$}}{\texttt{continue}}
                $e \gets (v_i^n, v_j^m)$\;
                $v_e^m \gets$ \texttt{getEarliestOutNode(e.group)} \label{lin: earlistNode}\;
                \If{\texttt{$e \in E_{pair}$}\label{line:e-in-pair}}
                {
                    \textcolor{blue}{\lIf{\texttt{StateStatus$(v_{e}^{m})\neq$'reachable'} \label{line:skip bi-pair}}
                    {\texttt{continue}}}
                    \textcolor{blue}{\texttt{UpdateStateStatus}(\texttt{StateStatus}, $\mathcal{G}$, $v_j^m$)\;}
                    
                    %$\tilde{e} \gets (v_{j+1}^n, v_{i-1}^m)$\;\label{line:get-pair}
                    % \lIf{$(v_{j+1}^n, v_{i-1}^m) \in E_{vis}$ \fbox{\textbf{or} $\exists v_{i'}^m \in V_{vis}: i' < i$}\label{line:skip-e}}
                    % {\texttt{continue}\label{line:skip-e2}}
                % \texttt{$visited \gets visited \cup \{e\}$}\; 
                }
                \textcolor{blue}{\lIf{\texttt{StateStatus$(v_{j-1}^n)\neq$'reachable'}\label{line:skip normal}}
                {\textcolor{blue}{\texttt{continue}}}}

                \lIf{\texttt{DeadlockCycleDetection}($\mathcal{G}$, $v_j^n$, $v_o$, $E_{vis} \cup \{e\}$, \texttt{StateStatus})}
                {\Return \texttt{true}}
            }
            \Return \texttt{false}\;
\end{algorithm}
\section{BTPG-max} \label{sec:btpg-max}
We propose BTPG-max, which, like BTPG-o, checks type-2 edges one by one and converts a type-2 edge into a bi-pair if adding the other edge in the bi-pair does not introduce deadlocks. The main difference between BTPG-o and  BTPG-max lies in cycle detection---BTPG-max considers only reachable deadlock cycles using \Cref{the: non-deadlock cycle} and \Cref{def: reachable deadlock cycle}. 

The main idea of cycle detection is to perform a depth-first search rooted at the edge we are checking to find a reachable deadlock cycle. To ensure that the cycles we identify meet the conditions of a reachable deadlock cycle, we propose requirements for traversing type-2 edges (normal or bi-pair) during cycle detection below. During cycle detection, we can omit those edges that do not satisfy the following requirements as those edges are guaranteed not in the current reachable deadlock cycle.

\subsection{Requirements of traversing a type-2 edge}

\subsubsection{Normal type-2 edges} 
According to \Cref{def: reachable deadlock cycle}, if type-2 edge \((v_i^n, v_j^m)\) is in a reachable deadlock cycle, then agent \(m\) must be in \(v_{j-1}^m\) when this deadlock occurs. Therefore, vertex \(v_{j-1}^m\) must be reachable for agent $m$.

% In \Cref{fig: new non-deadlock cycle}, the edge $(v_j^3,v_g^2)$ should not be traversed because node $v_{g-1}^2$ is unreachable.

\subsubsection{Bi-pair type-2 edges}
According to \Cref{the: non-deadlock cycle}, if edge \((v_i^n, v_j^m)\) is in a bi-pair, then in addition to vertex \(v_{j-1}^m\) being reachable, vertex \(v_{i-1}^n\) must also be reachable.
\vspace{-0.1em}
% In \Cref{fig: new non-deadlock cycle}, the edge $(v_g^2,v_c^1)$ should not be traversed because this edge is a type-2 edge in a bi-pair and the node $v_{g-1}^2$ is unreachable.
\subsection{Effects of traversing a type-2 edge}
After traversing a type-2 edge \((v_i^n, v_j^m)\) during cycle detection, meaning that we suppose that this edge is in a reachable deadlock cycle, it will lead to the following effect: All locations of \(v_k^m\) where \(k < j-1\) must have been visited and are considered as unreachable according to \Cref{def: reachable deadlock cycle}. In addition, based on \Cref{lem: unreachable}, all vertices that $v_j^m$ is connected to through only type-1 edges and normal type-2 edges are also unreachable. 

\begin{figure*}[t]
    \centering
    \vspace{-0.5em}
    \includegraphics[width=0.8\linewidth]{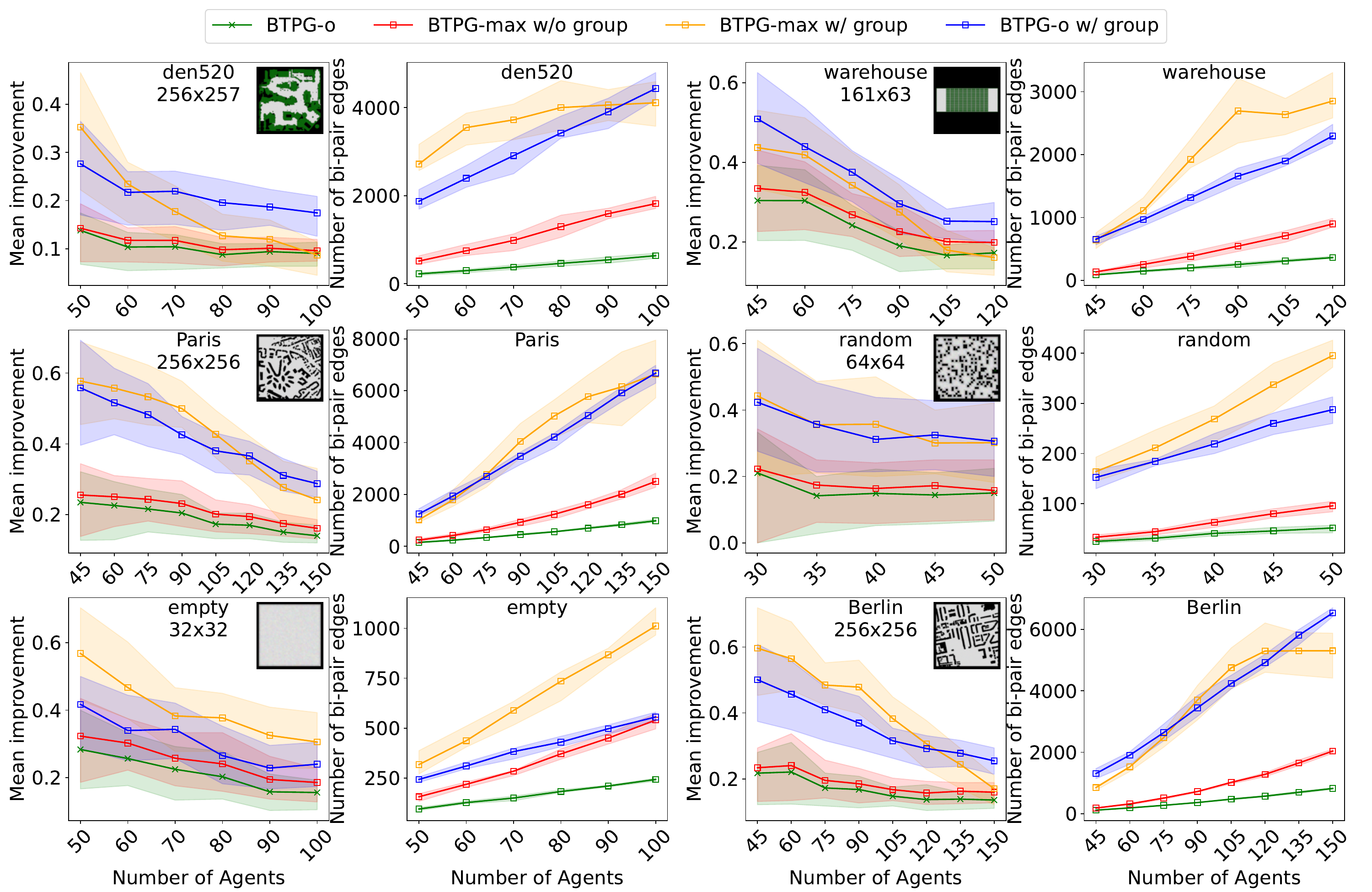}
    \vspace{-0.5em}
    \caption{Comparison of BTPG-o, BTPG-max w/o grouping, and BTPG-max w/ groups.}
    \label{fig:result}
    \vspace*{-1em}
\end{figure*}
\subsection{BTPG-max}
BTPG-max follows the general approach of BTPG construction: From a valid TPG, check each type-2 edge (\Cref{line:for-loop},\Cref{algo:Constructing BTPG - baseline}) to see if it forms a reachable deadlock cycle (\Cref{line:call-hasCycle}, \Cref{algo:Constructing BTPG - baseline}). If not, it is converted into a bi-pair. Since the construction approach being the same as BTPG-o, where each type-2 edge is checked one by one, BTPG-max inherits the anytime property of BTPG-o. This means that the algorithm can be stopped at any time, and the BTPG always remains valid. 

In BTPG-max, \texttt{StateStatus} records the status (reachable or unreachable) of each vertex in the BTPG and is used to check the requirements when traversing a type-2 edge. \texttt{UpdateStateStatus} will update \texttt{StateStatus} based on the effects of traversing a type-2 edge.

An important detail is that \texttt{UpdateStateStatus} can only pass through type-2 edges that have been checked not to be bi-pairs. In \Cref{fig: new non-deadlock cycle}, if edge $(v_d^3,v_b^2)$ is later converted into a bi-pair, even if this type-2 edge cannot be satisfied, according to the BTPG execution policy, agent 2 can choose the other direction and pass through node $v_b^2$ to reach $v_{c-1}^2$. In this case, the cycle $v_d^3 \rightarrow \cdots \rightarrow v_e^3 \rightarrow v_c^2 \rightarrow v_a^1 \rightarrow v_f^4 \rightarrow v_d^3$ becomes a deadlock cycle, reversing the previous result that allowed edge $(v_{d+1}^3,v_{f-1}^4)$ to be reversed.\footnote{Our appendix walks through BTPG-max on this example.}

Similar to BTPG-o, once a new bi-pair is found, it may allow previously non-convertible edges to be converted into bi-pairs because we may ignore more type-2 edges during cycle detection. Therefore, to find the locally maximal number of bi-pairs, a while loop is needed to repeatedly check if normal type-2 edges can be converted until no new bi-pairs are found. However, as mentioned above, transforming a normal type-2 edge into a bi-pair can also affect edges previously determined to be bi-pairs. Thus, in the second and subsequent iterations of this while loop, the condition for converting a normal type-2 edge into a bi-pair must include that it should not affect edges that were previously determined to be bi-pairs. However, in practical use, if this locally maximal property is not critical, one can implement a simpler version without the while loop. Empirically, we observed that the number of bi-pairs identified in subsequent iterations is usually very small and has a negligible impact on the results. We demonstrate this finding in Section \ref{sec:empirical-evaluation}. For clarity and simplicity, we have omitted the while loop in Algorithm 1's pseudocode.

\subsection{BTPG-max with grouping}
\label{sec:grouping}
BTPG-o identified two common groups of type-2 edges where each individual edge in the group cannot be converted into a bi-pair. The first group corresponds to agents following each other across consecutive locations (\Cref{fig:grouping} top row, agent 2 is following agent 1 across $BCD$), while the second group corresponds to agents crossing consecutive locations in opposite directions (\Cref{fig:grouping} bottom row, agent 1 crosses $BCD$ then agent 2 crosses $DCB$). Although individual edges cannot be reversed, all edges in the group can be reversed simultaneously \cite{berndt2023receding}. We incorporate this finding into our approach.
% For the two cases shown in \Cref{fig:grouping}, before reverse agent 1 passes through the corridor $BCD$ first, after reverse agent 2 passes through the corridor first. 

To determine if an edge group can be reversed, we need to reverse all edges in the group simultaneously and then verify that each one can be reversed. If any edge in the group cannot be reversed, the whole group cannot be reversed. 

In cycle detection, passing through an edge in one bi-pair edge group means that all edges in the group must be consistent with the direction it passes through. Therefore, we need to ensure that the position that allows the agent to choose this direction is reachable, i.e., the start of each group needs to be reachable. In \Cref{fig:grouping}, this corresponds to nodes B and D. More specifically, if the edges $(v_i^n, v_j^m), (v_{i+1}^n, v_{j\pm 1}^m), \cdots$ are in a bi-pair edge group and during cycle detection we want to pass through any one of them, then we need to check if the vertex $v_{i-1}^n$ (\Cref{algo: hasCycle}, \Cref{lin: earlistNode}) is reachable.

\subsection{Time Complexity}
% Algorithm 2 is a DFS which means it will have $O(E+V)$ recursive call.
\Cref{algo: hasCycle} is a recursive Depth-First Search (DFS), which visits each vertex and edge at most once and thus will be recursively called at most $O(|E|+|V|)$ times. Each recursive call has a time complexity of $O(|E|+|V|)$ due to \texttt{UpdateStateStatus} which is a separate DFS call. Thus, Alg 2 is $O((|E|+|V|)^2)$.

Now moving to \Cref{algo:Constructing BTPG - baseline}. \Cref{algo:Constructing BTPG - baseline} calls \texttt{DeadlockCycleDetection} a total of at most \( |E| \) times (since each edge needs to be checked).
Finally, \Cref{algo:Constructing BTPG - baseline} is invoked at most \( |E| \) times in the outer \texttt{while} loop (for repeatedly checking). 
Therefore, the total worst-case time complexity of the BTPG-max algorithm is:
$O((|E| + |V|)^2 \cdot |E|^2)$. We emphasize this is a worst-case time complexity and that in practice it can converge faster. Also, as mentioned earlier, BTPG-max is an anytime algorithm which allows it to be stopped early while still returning a valid BTPG.

\begin{figure}[t]
    \centering
    \includegraphics[width=0.85\linewidth]{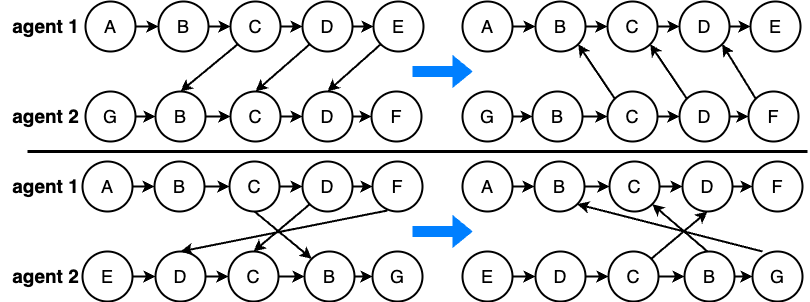}
    \caption{Two cases of grouping: before and after reversing.}
    \label{fig:grouping}
    \vspace*{-1em}
\end{figure}
\begin{table*}[th!]
    \centering
    \fontsize{9}{10}\selectfont
    % \resizebox{\textwidth}{!}{
    \setlength{\tabcolsep}{4pt}
    \begin{tabular}{|c|r|r|r|r|r|r|r|r|r|r|r|r|}
    
    \hline
    % & \multicolumn{4}{|c|}{Display Format}& \\ \cline{2-5}
    &\multicolumn{2}{|c|}{\texttt{den520}}& 
    \multicolumn{2}{|c|}{\texttt{warehouse}} &
    \multicolumn{2}{|c|}{\texttt{Paris}} &
    \multicolumn{2}{|c|}{\texttt{random}} &
    \multicolumn{2}{|c|}{\texttt{empty}} &
    \multicolumn{2}{|c|}{\texttt{Berlin}} \\ \hline
    & O& max& O& max& O&max& O&max& O&max& O&max \\ \hline
    Mean improvement & \textbf{9.1\%}& 8.9\%& \textbf{18.3\%}& 16.0\%& 15.5\%& \textbf{24.5\%} & 14.5\%& \textbf{32.7}\%& 19.8\%& \textbf{40.0\%}& 15.5\%& \textbf{17.2\%} \\ \hline
    Median improvement & \textbf{8.1\%}& 7.0\%& \textbf{18.3\%}& 15.2\%& 15.8\%& \textbf{24.4\%}& 11.8\%& \textbf{28.6\%}& 18.3\%& \textbf{29.3\%}& 15.3\%& \textbf{15.9\%} \\ \hline
    Max improvement & 21.6\%& \textbf{27.9\%}& 31.2\%& \textbf{38.3\%}& 26.3\%& \textbf{44.2\%}& 63.6\%& \textbf{76.9\%}& 42.8\%& \textbf{60.0\%}& 24.6\%& \textbf{40.3\%} \\ \hline
    Min improvement & \textbf{2.2}\%& 0.9\%& \textbf{8.8}\%& 5.1\%& \textbf{6.8}\%& 4.9\%& 0.0\%& \textbf{5.8\%}& 5.0\%& \textbf{6.3\%}& \textbf{7.0}\%& 1.0\% \\ \hline \hline
     \# Type-2 edges& \multicolumn{2}{|c|}{25,766} & \multicolumn{2}{|c|}{15,228} & \multicolumn{2}{|c|}{24,506} & \multicolumn{2}{|c|}{1,153} & \multicolumn{2}{|c|}{3,044} & \multicolumn{2}{|c|}{26,025} \\ \hline
    % \# Singleton edges & \multicolumn{2}{|c|}{2,172} & \multicolumn{2}{|c|}{1,051} & \multicolumn{2}{|c|}{2,688} & \multicolumn{2}{|c|}{146} & \multicolumn{2}{|c|}{817} & \multicolumn{2}{|c|}{2,265} \\ \hline 
    \# Bi-Pairs found& 1,041& \textbf{4,110}& 532& \textbf{2,851}& 1,679& \textbf{6,649}& 69& \textbf{395}& 360& \textbf{1,011}& 1,389& \textbf{5,303} \\ \hline
    \# Used Bi-Pairs & \textbf{68}& 62& \textbf{56}& 55& 120& \textbf{169}& 5.6& \textbf{14}& 38& \textbf{51}& 110& \textbf{128}\\ \hline 
    % \hline
    % BTPG runtime (s)& 23.72& 161.88& 6.31& 8.72& 58.12& 181.33& 0.03& 0.04& 0.69& 1.06& 43.32& 72.87 \\ \hline
    % MAPF runtime (s)& \multicolumn{2}{|c|}{1.93} & \multicolumn{2}{|c|}{9.44} & \multicolumn{2}{|c|}{1.12} & \multicolumn{2}{|c|}{19.80} & \multicolumn{2}{|c|}{1.11} & \multicolumn{2}{|c|}{9.90} \\ \hline

    % $\#$ of State& 18326& 27912& 10724& 1373& 27912& 10724 \\ \hline
    % $\#$ of Type2& 30327& 23943& 14696& 1373& 23943& 30327 \\ \hline
    
    \end{tabular}
    
   \caption{Statistics of BTPG-o and BTPG-max w/ grouping. The number of agents for the six maps selected for the statistics are the largest ones, respectively. o: BTPG-optimized; max: BTPG-max w/ grouping. All data in the bottom block are averages of 10 scenarios for each map. Used Bi-Pairs: Bi-Pairs that are used (reversed) in the simulation.
   }
   \label{tab:Statistics of BTPG}
   \vspace*{-1.2em}
\end{table*}

\section{Empirical Evaluation} \label{sec:empirical-evaluation}

% \subsubsection{Experiment settings}
We use the optimal MAPF solver CBSH2-RTC \cite{li_pairwise_2021} to generate the MAPF solution for each MAPF instance. We then convert each MAPF solution into a TPG. We run our proposed method BTPG-max with and without grouping and compare them to the baseline BTPG-o. All BTPG methods have a 10-minute time cutoff. Finally, we simulate the execution policies of TPG and BTPG, where 10\% of the agents have a 30\% chance of being delayed by 5 timesteps at each non-delayed timestep, same as the settings used in the BTPG-o paper~\cite{su2024bidirectional_tpg} for a fair comparison.

We ran a total of 3,900 simulations across six different benchmark maps (\Cref{fig:result}). Each map features 5 to 8 different agent counts, with 10 random instances per agent count. To simulate delays, each instance was run with 10 random seeds. The maximum number of agents per map is the highest number that CBSH2-RTC can solve within a 2-minute time limit. All experiments were run on a PC with a 3.40 GHz Intel i7-14700K CPU and 64 GB RAM. 

\begin{figure}[t]
    \centering
    \vspace{-1.0em}
    \includegraphics[width=0.65\linewidth]{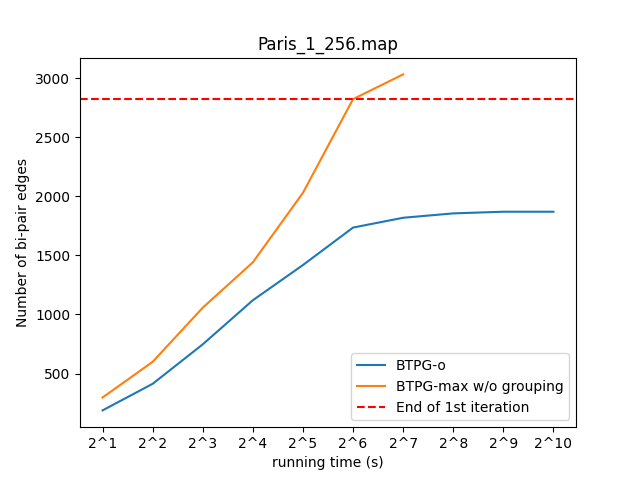}
    \vspace{-1.0em}
    \caption{Anytime behavior of BTPG-o/max w/o grouping in the map \texttt{Paris\_1\_256} with 150 agents.}
    \label{fig:anytime}
    \vspace{-1em}
\end{figure}

\subsubsection{Improvement}
We adopt the improvement calculation from~\citet{su2024bidirectional_tpg}, $ \frac{T_{TPG}-T_{BTPG}}{T_{TPG}-T_{Ideal}}$. $T_{TPG}$ and $T_{BTPG}$ represent the mean execution times for TPG and BTPG, respectively, corresponding to the average number of timesteps agents take to reach their target locations. $T_{Ideal}$ is the average lower bound of execution time, calculated as the sum of execution timesteps specified in the original MAPF solution and the total delay timesteps of the delayed agents, divided by the number of agents.

\Cref{fig:result} plots the results of BTPG-o and BTPG-max with and without grouping. Comparing BTPG-o without grouping (green) and BTPG-max without grouping (red), we see that for all maps BTPG-max is able to find more bi-pairs and has a higher mean improvement than BTPG-o. The performance with grouping shows a different story. Comparing BTPG-o without and with grouping (green vs blue), and BTPG-max without and grouping (red vs yellow), we see the advantage of grouping as it significantly improves both the number of bi-pairs found and the overall improvement. 

Interestingly, we see that the number of bi-pairs found and the mean improvement are not necessarily related. For example, in the warehouse map BTPG-max with grouping finds more bi-pairs than BTPG-o with grouping, but BTPG-o with grouping has a larger mean improvement. \Cref{tab:Statistics of BTPG}, which compares BTPG-o and BTPG-max with grouping, shows a similar story. BTPG-max consistently finds 3-5x more bi-pairs on all 6 maps but only has better median improvements on 4 of the 6 maps. 
We believe this discrepancy is due to the order of type-2 edges that BTPG-o picks. BTPG-o checks type-2 edges in the order they are encountered during execution, while BTPG-max converts more earlier edges into later edges. During execution though, the locations further back on the paths are more likely to be affected by delays. Therefore, in the \texttt{den520} and \texttt{warehouse} map (\Cref{tab:Statistics of BTPG}), even though the number of bi-pairs is larger, the number of used bi-pairs is smaller, so the improvement is smaller. 
% In larger and denser maps like \texttt{den520} and \texttt{warehouse}, as the number of agents increases, the number of bi-pairs found by BTPG-max decreases as the 10-minute cutoff is too small. This leads to a lower improvement compared to BTPG-o, as shown in the first row of \Cref{tab:Statistics of BTPG}.

% Additionally, because BTPG-o checks type-2 edges in the order they are encountered during execution, the type-2 edges encountered earlier are checked first to see if they can be converted into bi-pairs. Given a sufficiently large runtime where BTPG-max can run until saturation, this ordering does not affect performance. However, when the number of agents increases, the 10-minute time cutoff results in BTPG-max converting more earlier edges into bi-pairs rather than later edges. During execution though, the locations further back on the paths are more likely to be affected by delays. Therefore, in the \texttt{den520} and \texttt{warehouse} map (\Cref{tab:Statistics of BTPG}), even though the number of bi-pairs is larger, the number of used bi-pairs is smaller, so the improvement is smaller.

\subsubsection{Anytime Performance}
\Cref{fig:anytime} shows the anytime performance of BTPG-o vs BTPG-max without grouping on Paris\_1\_256 with 150 agents. The x-axis is the BTPG creation cut-off time in seconds while the y-axis is the number of bi-pairs found. We see that BTPG-max has strictly superior anytime performance than BPTG-o; BTPG-max finds more bi-pairs faster and saturates at a higher final value. 

The red horizontal dashed line denotes the end of the 1st iteration (after we have gone through all the type-2 edges once) and the start of the next iterations. We see that the vast majority of edges in BTPG-max are found during this first iteration with very few found afterwards.

Overall our results show that BTPG-max w/o grouping is strictly better (in bi-pairs found, mean improvement, anytime performance) than BTPG-o w/o grouping. Also, we find that grouping has a significant impact on performance and augments both BTPG-o and BTPG-max significantly.

\section{Conclusion}
Our work introduces BTPG-max, a method for identifying a locally maximal number of bi-pairs in BTPGs. Our key theoretical contribution is defining all possible non-deadlock cycles in BTPGs, provided in \Cref{the: non-deadlock cycle}, which lays the foundation for distinguishing deadlock cycles in BTPG construction. Building on this, BTPG-max checks all possible deadlock cycles to construct locally maximal BTPGs to maximize execution flexibility. Empirically, BTPG-max consistently discovers more bi-pairs than BTPG-o and, combined with the grouping strategy, yields significantly better execution performance.

There are several possible future directions. One could attempt to make BTPG-max an online method that is called during execution or improve the anytime performance of BTPG-max. Additionally, since checking type-2 edges is a semi-independent operation, a parallel version of BTPG-max would potentially be promising. Given the impact of grouping, more sophisticated methods of grouping could lead to further benefits. Finally, it would also be interesting work to try to incorporate prior knowledge of delays within the BTPG framework.
% \newpage

\bibliography{aaai2026}

% Check whether the conference requires a reproducibility checklist to be included in the paper.
% If so, you can uncomment the following line and ajust the path to include it.
% \input{checklist}
% \input{ReproducibilityChecklist}

% \clearpage
\appendix
\begin{figure*}[th!]
    \centering
    \includegraphics[width=\linewidth]{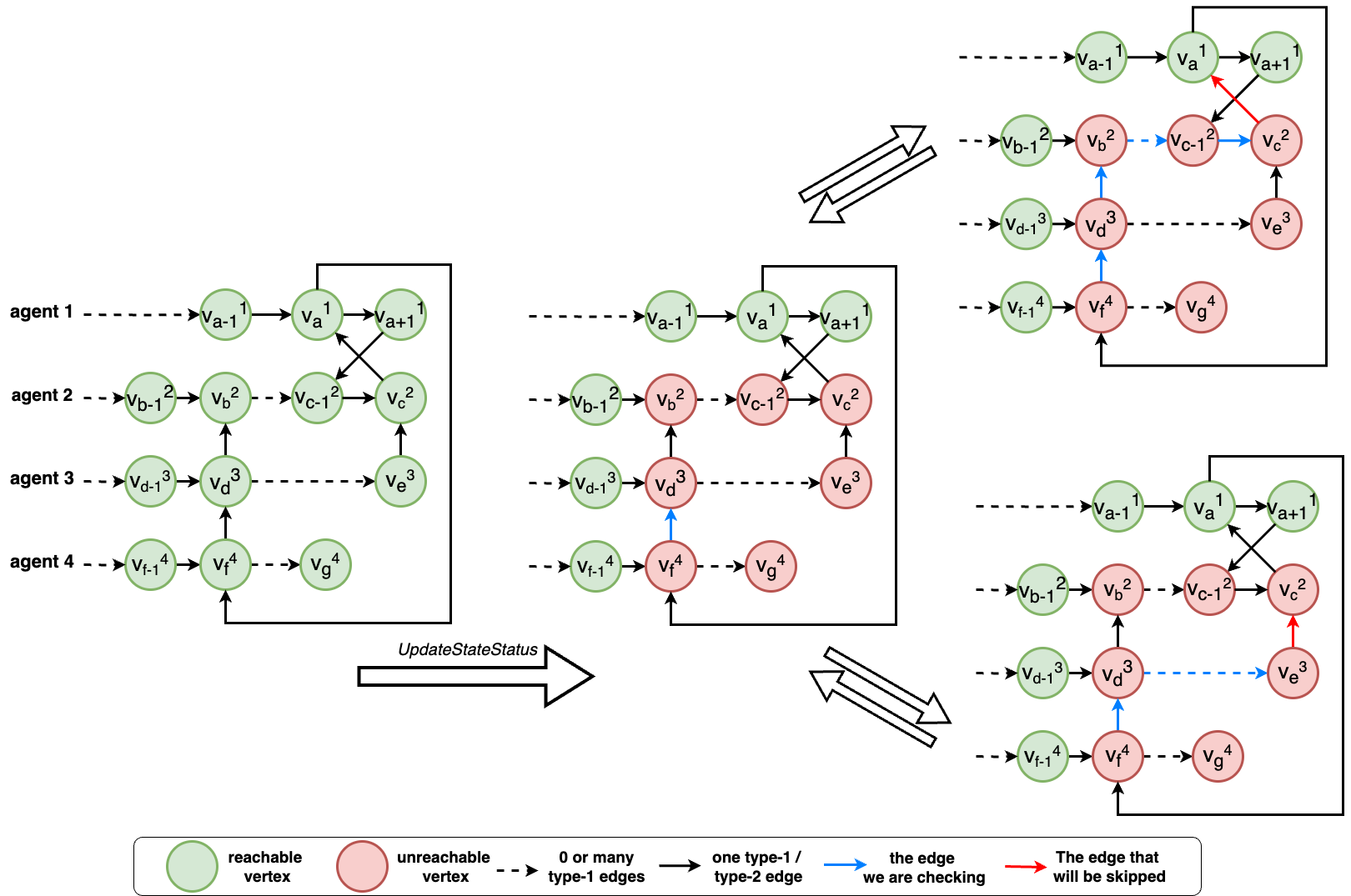}
    \caption{BTPG-max example}
    \vspace{-1em}
    \label{fig:btpg-max}
\end{figure*}
\section{Run-through of BTPG-max}
We will demonstrate the entire process of BTPG-max through \Cref{fig:btpg-max}.
In \Cref{algo: hasCycle}, \texttt{StateStatus} is used to track the conditions of each vertex on the map. At first, every vertex on the map is reachable (\Cref{lin:initial}, \Cref{algo:Constructing BTPG - baseline}). If we want to check if the reversed edge $(v_f^4,v_d^3)$ is valid, first the \texttt{StateStatus} should be updated (\Cref{lin:update},\Cref{algo:Constructing BTPG - baseline}). If the edge $(v_f^4, v_d^3)$ is in the deadlock cycle, then vertices $\{v_x^3|x \ge d\}$ should be unreachable, and vertices $\{v_x^3|x < d-1\}$ should be set to ``unreachable” based on the effect of traversing a type-2 edge. Then, since $v_d^3$ cannot be reached and the edge $(v_d^3,v_b^2)$ is a normal type-2 edge and cannot be satisfied, the vertices $\{v_x^2|x \ge b\}$ are also unreachable. \texttt{UpdateStateStatus} uses simple DFS to update the conditions of each vertice based on the effects of traversing a type-2 edge.

For the top right case, after passing through vertices $v_d^3$, $v_b^2$ and $v_{c-1}^2$, because vertex $v_{c-1}^2$ is unreachable and the edge $(v_c^2,v_a^1)$ is an edge in a bi-pair, BTPG-max will skip this edge (\Cref{line:skip bi-pair}, \Cref{algo: hasCycle}) based on the requirements of traversing a type-2 edge. Similarly, for the below right one, when traversing the vertices $v_d^3$ and $v_e^3$, since the node $v_{c-1}^2$ is unreachable, the edge $(v_e^3,v_c^2)$ will be skipped ( \Cref{line:skip normal}, \Cref{algo: hasCycle}).

\end{document}